\documentclass[pra,aps,twocolumn,showpacs]{revtex4}

\usepackage{amsmath,amsfonts,amssymb,caption,hyperref,color,epsfig,graphics,graphicx,latexsym,mathrsfs,revsymb,theorem,url,epstopdf}

\hypersetup{colorlinks,linkcolor={blue},citecolor={blue},urlcolor={red}}

\usepackage{hyperref}

\newlength{\halfpagewidth}

\divide\halfpagewidth by 2

\newtheorem{definition}{Definition}
\newtheorem{proposition}[definition]{Proposition}
\newtheorem{Lemma}[definition]{Lemma}

\newtheorem{Theorem}[definition]{Theorem}
\newtheorem{Corollary}[definition]{Corollary}
\newtheorem{conjecture}[definition]{Conjecture}

\newtheorem{remark}[definition]{Remark}
\newtheorem{example}[definition]{Example}
\newtheorem{question}[definition]{Question}

\def\squareforqed{\hbox{\rlap{$\sqcap$}$\sqcup$}}
\def\qed{\ifmmode\squareforqed\else{\unskip\nobreak\hfil
		\penalty50\hskip1em\null\nobreak\hfil\squareforqed
		\parfillskip=0pt\finalhyphendemerits=0\endgraf}\fi}
\def\endenv{\ifmmode\;\else{\unskip\nobreak\hfil
		\penalty50\hskip1em\null\nobreak\hfil\;
		\parfillskip=0pt\finalhyphendemerits=0\endgraf}\fi}
\newenvironment{proof}{\noindent \textbf{{Proof.~} }}{\qed}
\def\Dbar{\leavevmode\lower.6ex\hbox to 0pt
	{\hskip-.23ex\accent"16\hss}D}
\makeatletter
\def\url@leostyle{%
	\@ifundefined{selectfont}{\def\UrlFont{\sf}}{\def\UrlFont{\small\ttfamily}}}
\makeatother
\def\bcj{\begin{conjecture}}
	\def\ecj{\end{conjecture}}
\def\bcr{\begin{corollary}}
	\def\ecr{\end{corollary}}
\def\bd{\begin{definition}}
	\def\ed{\end{definition}}
\def\bea{\begin{eqnarray}}
\def\eea{\end{eqnarray}}
\def\bem{\begin{enumerate}}
	\def\eem{\end{enumerate}}
\def\bex{\begin{example}}
	\def\eex{\end{example}}
\def\bim{\begin{itemize}}
	\def\eim{\end{itemize}}
\def\bl{\begin{lemma}}
	\def\el{\end{lemma}}
\def\bma{\begin{bmatrix}}
	\def\ema{\end{bmatrix}}
\def\bpf{\begin{proof}}
	\def\epf{\end{proof}}
\def\bpp{\begin{proposition}}
	\def\epp{\end{proposition}}
\def\bqu{\begin{question}}
	\def\equ{\end{question}}
\def\br{\begin{remark}}
	\def\er{\end{remark}}
\def\bt{\begin{theorem}}
	\def\et{\end{theorem}}

\def\btb{\begin{tabular}}
	\def\etb{\end{tabular}}

\newcommand{\nc}{\newcommand}

\def\r{\rho}

\def\ps{\psi}

\nc{\bbA}{\mathbb{A}} \nc{\bbB}{\mathbb{B}} \nc{\bbC}{\mathbb{C}}
\nc{\bbD}{\mathbb{D}} \nc{\bbE}{\mathbb{E}} \nc{\bbF}{\mathbb{F}}
\nc{\bbG}{\mathbb{G}} \nc{\bbH}{\mathbb{H}} \nc{\bbI}{\mathbb{I}}
\nc{\bbJ}{\mathbb{J}} \nc{\bbK}{\mathbb{K}} \nc{\bbL}{\mathbb{L}}
\nc{\bbM}{\mathbb{M}} \nc{\bbN}{\mathbb{N}} \nc{\bbO}{\mathbb{O}}
\nc{\bbP}{\mathbb{P}} \nc{\bbQ}{\mathbb{Q}} \nc{\bbR}{\mathbb{R}}
\nc{\bbS}{\mathbb{S}} \nc{\bbT}{\mathbb{T}} \nc{\bbU}{\mathbb{U}}
\nc{\bbV}{\mathbb{V}} \nc{\bbW}{\mathbb{W}} \nc{\bbX}{\mathbb{X}}
\nc{\bbZ}{\mathbb{Z}}

\nc{\bA}{{\bf A}} \nc{\bB}{{\bf B}} \nc{\bC}{{\bf C}}
\nc{\bD}{{\bf D}} \nc{\bE}{{\bf E}} \nc{\bF}{{\bf F}}
\nc{\bG}{{\bf G}} \nc{\bH}{{\bf H}} \nc{\bI}{{\bf I}}
\nc{\bJ}{{\bf J}} \nc{\bK}{{\bf K}} \nc{\bL}{{\bf L}}
\nc{\bM}{{\bf M}} \nc{\bN}{{\bf N}} \nc{\bO}{{\bf O}}
\nc{\bP}{{\bf P}} \nc{\bQ}{{\bf Q}} \nc{\bR}{{\bf R}}
\nc{\bS}{{\bf S}} \nc{\bT}{{\bf T}} \nc{\bU}{{\bf U}}
\nc{\bV}{{\bf V}} \nc{\bW}{{\bf W}} \nc{\bX}{{\bf X}}
\nc{\bZ}{{\bf Z}}

\nc{\cA}{{\cal A}} \nc{\cB}{{\cal B}} \nc{\cC}{{\cal C}}
\nc{\cD}{{\cal D}} \nc{\cE}{{\cal E}} \nc{\cF}{{\cal F}}
\nc{\cG}{{\cal G}} \nc{\cH}{{\cal H}} \nc{\cI}{{\cal I}}
\nc{\cJ}{{\cal J}} \nc{\cK}{{\cal K}} \nc{\cL}{{\cal L}}
\nc{\cM}{{\cal M}} \nc{\cN}{{\cal N}} \nc{\cO}{{\cal O}}
\nc{\cP}{{\cal P}} \nc{\cQ}{{\cal Q}} \nc{\cR}{{\cal R}}
\nc{\cS}{{\cal S}} \nc{\cT}{{\cal T}} \nc{\cU}{{\cal U}}
\nc{\cV}{{\cal V}} \nc{\cW}{{\cal W}} \nc{\cX}{{\cal X}}
\nc{\cZ}{{\cal Z}}

\nc{\hA}{{\hat{A}}} \nc{\hB}{{\hat{B}}} \nc{\hC}{{\hat{C}}}
\nc{\hD}{{\hat{D}}} \nc{\hE}{{\hat{E}}} \nc{\hF}{{\hat{F}}}
\nc{\hG}{{\hat{G}}} \nc{\hH}{{\hat{H}}} \nc{\hI}{{\hat{I}}}
\nc{\hJ}{{\hat{J}}} \nc{\hK}{{\hat{K}}} \nc{\hL}{{\hat{L}}}
\nc{\hM}{{\hat{M}}} \nc{\hN}{{\hat{N}}} \nc{\hO}{{\hat{O}}}
\nc{\hP}{{\hat{P}}} \nc{\hR}{{\hat{R}}} \nc{\hS}{{\hat{S}}}
\nc{\hT}{{\hat{T}}} \nc{\hU}{{\hat{U}}} \nc{\hV}{{\hat{V}}}
\nc{\hW}{{\hat{W}}} \nc{\hX}{{\hat{X}}} \nc{\hZ}{{\hat{Z}}}

\nc{\hn}{{\hat{n}}}

\def\dim{\mathop{\rm Dim}}

\def\max{\mathop{\rm max}}
\def\min{\mathop{\rm min}}

\def\supp{\mathop{\rm supp}}
\def\tr{\mathop{\rm Tr}}

\newcommand{\bra}[1]{\langle#1|}
\newcommand{\ket}[1]{|#1\rangle}

\def\Dbar{\leavevmode\lower.6ex\hbox to 0pt
	{\hskip-.23ex\accent"16\hss}D}

\begin{document}
	\title{An Extention of Entanglement Measures for Pure States}
	
	\author{Xian Shi}\email[]
	{shixian01@buaa.edu.cn}
	\affiliation{School of Mathematical Sciences, Beihang University, Beijing 100191, China}
	\author{Lin Chen}\email[]{linchen@buaa.edu.cn (corresponding author)}
	\affiliation{School of Mathematical Sciences, Beihang University, Beijing 100191, China}
	\affiliation{International Research Institute for Multidisciplinary Science, Beihang University, Beijing 100191, China}
	
	%
	
	
	
	\date{\today}
	

\begin{abstract}
To quantify the entanglement is one of the most important topics in quantum entanglement theory. In [arXiv: 2006.12408], the authors proposed a method to build a measure from the orginal domain to a larger one. Here we apply that method to build an entanglement measure from measures for pure states. First, we present conditions when the entanglement measure is an entanglement monotone and convex, we also present an interpretation of the smoothed one-shot entanglement cost under the method here. At last, we present a difference between  the local operation and classical communication (LOCC) and the separability-preserving (SEPP) operations, then we present the entanglement measures built from the geometric entanglement measure for pure states by the convex roof extended method and the method here are equal, at last, we present the relationship between the concurrence and the entanglement measure built from concurrence for pure states by the method here on $2\otimes 2$ systems. We also present the measure is monogamous for $2\otimes2\otimes d$ system.
\end{abstract}

\maketitle
\section{introduction}
\indent Quantum entanglement is one of the essiential features in quantum mechanics when comparing with the classical physics \cite{horodecki2009quantum,plenio2014introduction}. It also plays key roles in quantum information processing, such as, quantum cryptography \cite{ekert1991quantum}, quantum teleportation \cite{bennett1993teleporting} and quantum superdense coding \cite{bennett1992communication}. \\
\indent One of the most important and interesting problems in studying the entanglement is how to quantify the entanglement in a composite quantum system. In 1996, the authors in \cite{bennett1996mixed} proposed the distillable entanglement and entanglement cost and presented their operational interpretations. The authors in \cite{vedral1997quantifying} presented three necessary conditions that an entanglement measure should satisfy in 1997, and one of the important conditions is that the quantum entanglement cannot increase under LOCC. In 2000, Vidal proposed a general mathematical framework for entanglement measures \cite{vidal2000entanglement}. There the author also presented a convex roof extended method to bulid entanglement monotone for bipartite entangled systems from some functions on bipartite pure states. The other important method to quantify the quantum entanglement is based on the distance to the closest separable state. The most important examples are the geometric measures \cite{wei2003geometric,markham2007entanglement} and the quantum relative entropy \cite{horodecki2005local}. Due to the monotonicity of the inner product and quantum relative entropy under the LOCC, it is clear that the above two are entanglement measures.  Another important method to build an entanglement measure of a bipartite state $\rho_{AB}$ is the minimum quantum conditional mutual information of all the extensions of $\rho_{AB}.$ there the authors named the measure the squashed entanglement \cite{christandl2004squashed}. Compared with the entanglement distillation, the squashed entanglement is additive on tensor products and superadditive in general.  Recently, Gour and Tomamichel proposed a new method to quantify the resource for the general resources \cite{gour2020optimal}.\\
\indent  In this paper, we mainly apply the method to build an entanglement measure for mixed states from the measures for pure states. Given an entanglement measure $E$ for pure states in bipartite systems, we first present a sufficient condition when an entanglement measure built from the method here is an entanglement monotone. Then we consider the relation between an entanglement measure built from the method here and the convex roof extended method \cite{vidal2000entanglement}. And we also present a condition when the entanglement measure built from the method here is convex. As an application, we present a difference between the LOCC and SEPP by the Schmidt number under the method here \cite{terhal2000schmidt,gour2020optimal}, this is an entanglement measure which can be increased under the separability-preserving operations. Then we present the relation between the geometric entanglement measure under the convex roof extend and the method proposed here, we also present that for $2\otimes 2$ systems, the concurrence built from the convex roof extended method and the method here are equal, based on the result, we have that the concurrence satisfies the monogamy of entanglement proposed in \cite{gour2018monogamy} for $2\otimes2\otimes d$ systems.\\
\indent This paper is organized as follows. In Sec. \MakeUppercase{\romannumeral2}, we first present the preliminary knowledge needed here, and then we present a sufficient condition when the entanglement measure built here is entanglement monotone, we also consider a condition when the entanglement measure built from the method here is convex. In Sec. \MakeUppercase{\romannumeral3}, we present an interpretation of the smoothed one-shot entanglement cost under the method here. In Sec. \MakeUppercase{\romannumeral4}, we present some applications of the entanglement measure bulit from here, first we present a difference between the LOCC operations and the separability-preserving operations, then we present the relationship between the geometric entanglement measure for pure states built from the convex roof extended method and the method here, at last, we consider the entanglement measure generalized from the concurrence for pure states in $2\otimes2$ systems, and then we show the entanglement measure is monogamous for $2\otimes2\otimes d$ system. In Sec. \MakeUppercase{\romannumeral5}, we end with a conclusion.
\section{Building entanglement monotone for pure states}
\indent In this section, first we recall some preliminary knowledge on the entanglement measures and operations of entanglement theory. Then based on the method \cite{gour2020optimal} proposed , we will propose some entanglement measures built from entanglement measure for pure states, and we present some sufficient conditions when entanglement measures are entanglement monotone, convex and subadditivity.\\
\indent In the following, we denote $\mathcal{D}(\mathcal{H}_{AB})$ the set of states on $\mathcal{H}_{AB}$ and $\mathcal{S}(\mathcal{H}_{AB})$ the set of separable states on $\mathcal{H}_{AB}.$ If a bipartite pure state $\ket{\psi}_{AB}$ can be written as $\ket{\psi}_{AB}=\ket{\phi_1}_A\otimes\ket{\phi_2}_B,$ then $\ket{\psi}_{AB}$ is separable, otherwise, $\ket{\psi}_{AB}$ is an entangled state. If a mixed state $\rho_{AB}$ can be written as $\rho_{AB}=\sum_i p_i\rho_A^i\otimes\sigma_B^i,$ then $\rho_{AB}$ is separable, otherwise, $\rho_{AB}$ is entangled. \\
\indent Recall that an entanglement measure $E:\mathcal{D}(\mathcal{H}_{AB})\rightarrow\mathbb{R}^{+}$ is an entanglement measure \cite{vedral1997quantifying} if it satifies:
\begin{itemize}
	\item [(i)] $E(\rho_{AB})=0$ if $\rho_{AB}\in \mathcal{S}(\mathcal{H}_{AB})$
	\item [(ii)]$E$ doesnot increase under the LOCC operation.\begin{align*}
	E(\Psi(\rho_{AB}))\le E(\rho_{AB}),
	\end{align*} 
	here $\Psi$ is an LOCC operation. 
\end{itemize}\par
In \cite{vidal1999entanglement}, the author presented that when $E$ satisfies the following two conditions, $E$ is an entanglement monotone,
\begin{itemize}
	\item [(iii)]$E(\rho)\ge \sum_k p_kE(\sigma_k),$ here $\sigma_{k}=\frac{\mathcal{E}_{i,k}(\rho_{AB})}{p_k},$ $p_k=\tr\mathcal{E}_{i,k}(\rho_{AB})$, $\mathcal{E}_{i,k}$ is any unilocal quantum operation performed by any party $A$ or $B.$
	\item [(iv)] For any decomposition $\{p_k,\rho_k\}$ of $\rho_{AB}$
	\begin{align*}
	E(\rho)\le \sum_k p_kE(\rho_k)
	\end{align*}
\end{itemize}
\par Obviously, when $E$ is an entanglement monotone, $E$ is an entanglement measure.\\
\indent Assume $\ket{\psi}_{AB}$ is a bipartite pure state in $\mathcal{H}_{AB}$ that can be written as $\ket{\psi}_{AB}=\sum_i\sqrt{\lambda_i}\ket{ii},$ and let $\mu^{\downarrow}(\ket{\psi}_{AB})$ be the vector $(\lambda_0,\lambda_1,\cdots,\lambda_{d-1})$ in decreasing order, then we recall entanglement measures $E_k$ for pure states $\ket{\psi}_{AB},$ $E_k(\psi_{AB})=f_k(\tr\ket{\psi}_{AB}\bra{\psi})=\sum_{i=k-1}^{d-1}\lambda_i$ \cite{vidal2000entanglement}. Next we recall that if an entanglement measure $E$ for a pure state $\ket{\psi}$ is the same entanglement ordering \cite{virmani2000ordering} with $E_k,$ $k=1,2,\cdots,d,$ we mean that if for any two vectors $\ket{\psi_1}$ and $\ket{\psi_2},$ $E_k(\psi_1)\ge E_k(\psi_2),k=1,2,\cdots,d,$ then $E(\psi_1)\ge E(\psi_2).$\\
\indent Assume $\Lambda:A\rightarrow A^{'}$ is a completely positive and trace-preserving map, then its Choi matrix is $J_{\Lambda}=(I\otimes\Lambda)(\ket{\Psi}_{AA^{'}}{\bra{\Psi}}),$ here $\ket{\Psi}_{AA^{'}}$ is a maximally entangled state. As the LOCC operations are hard to characterise mathematically, then some important problems on quantum entanglement theory are hard to solve. Some meaningful methods proposed are to extend the set of LOCC operations \cite{rains1997entanglement,rains1999bound,rains2001semidefinite,eggeling2001distillability,brandao2010reversible}, which makes some problems on distinguishing and transformation of entangled states much easier to solver. Then we propose the structures of separable operations (SEP), positive partial transpose (PPT) operations, and separability-preserving (SEPP) operations,
\begin{align*}
SEP=&\{\Lambda|\Lambda=\sum_i (A_i\otimes B_i)^{\dagger}\cdot(A_i\otimes B_i)\}\\
PPT=&\{\Lambda|J_{\Lambda}^{T_{BB^{'}}}\ge 0\}\\
SEPP=&\{\Lambda|\rho \textit{ is separable}\Longrightarrow \Lambda(\rho) \textit{is separable.}\}
\end{align*}
\indent Recently, Gour and Tomamichel proposed a new method to extend the resource measures from one domain to a larger one \cite{gour2020optimal}. Yu $et$ $al.$ considered the coherence measures in terms of the method and presented operational interpretations for some coherence measures \cite{yu2020quantifying}. Here we apply this method to the entanglement theory to present new entanglement measures, and then we consider the properties of the entanglement measures.\\
\indent Assume $\ket{\psi}_{AB}$ is a pure state in $\mathcal{H}_{AB}$, $E$ is an entanglement measure for pure states in $\mathcal{H}_{AB},$  then we extend the above measure for pure states to a corresponding quantity for the mixed states,
\begin{align}
\overline{E}(\rho_{AB})=\inf_{\ket{\psi}_{AB}\in \mathcal{R}(\rho_{AB})} E(\ket{\psi}_{AB}),\label{Ev} 
\end{align} 
here the infimum takes over all the pure states in the set $\mathcal{R}(\rho_{AB})=\{ \psi_{AB}\in \mathcal{H}_{AB}|\rho_{AB}=\Lambda(\psi_{AB}),\Lambda\in \mathcal{T}.\}$ Here $\mathcal{T}$ stands for LOCC, SEP, PPT or SEPP.\par
 Next we recall the convex roof extended method to bulid an entanglement monotone for a mixed state that Vidal proposed in \cite{vidal1999entanglement}. \par 
Assume $E(\ket{\psi}_{AB})=f(\tr_B(\ket{\psi}_{AB}\bra{\psi})), f:\mathcal{D}(\mathcal{H}_A)\rightarrow\mathcal{R}^{+}.$ If $f$ satisfies the following conditions:\\
\begin{itemize}
	\item[(i)] $U$-invariant: $f(U\sigma U^{\dagger})=f(\sigma),$ $\forall \sigma\in \mathcal{D}(\mathcal{H}_A), U$ is a unitary matrix on $\mathcal{H}_A,$
	\item[(ii)] concave: $f(\lambda\sigma_1+(1-\lambda)\sigma_2)\ge \lambda_1f(\sigma)+(1-\lambda)f(\sigma_2),$ here $\sigma_i\in \mathcal{D}(\mathcal{H}_A),$ $i=1,2,$ $\lambda\in (0,1).$ 
\end{itemize}
Vidal showed $E$ is an entanglement monotone for mixed states by the convex roof extended method \cite{vidal1999entanglement},
\begin{align}
E_f(\rho_{AB})=\min_{\{p_i,\ket{\psi}_{AB}\}}\sum_ip_iE(\ket{\psi_i}),\label{cre}
\end{align}
where the minimum takes over all the decomposition of $\{p_i,\ket{\psi_i}_{AB}\}$ such that $\rho_{AB}=\sum_ip_i\ket{\psi_i}\bra{\psi_i}.$\\
\begin{Theorem}\label{th1}
	Assume $\rho_{AB}\in \mathcal{D}(\mathcal{H}_{AB}),$ then when $\mathcal{T}$ consists of the LOCC operations, the entanglement measure defined as $(\ref{Ev})$ is an entanglement measure. If $E$ is entanglement monotone for pure states, then $\overline{E}$ satisfies the condition (iii),  when the function $f$ corresponding to $E$ satisfies $f(\lambda_1\Lambda_1+\lambda_2\Lambda_2)\le \lambda_1f(\Lambda_1)+\lambda_2f(\Lambda_2),$ here $\Lambda_i$, $i=1,2$ are diagonal matrices on the space $\mathcal{H}_A$,  then $\overline{E}$ satisfies the condition (iv).
\end{Theorem}
\begin{proof}
	The proof that (\ref{Ev}) is an entanglement measure can be found in \cite{gour2020optimal}.\par
	Then we prove the condition (iii) when $E$ is an entanglement monotone for pure states.\par
	If $\ket{\psi}_{AB}$ is a pure state. There exists a decomposition $\{p_k,\ket{\phi_k}\}$ such that $\ket{\psi}\stackrel{LOCC}{\rightarrow}\{p_k,\ket{\phi_k}\}$ of $\rho,$ then by the assumption,
	\begin{align}
	\overline{E}(\ket{\psi})\ge \sum_kp_k\overline{E}(\ket{\phi_k}),\label{0mu}
	\end{align}
\par Next when $\rho$ is a mixed state, assume $\ket{\psi}$ is the optimal pure state for $\rho$ in terms of $E$, there exists a decomposition $\{r_j,\ket{\eta_j}\}$ of $\rho$ such that
\begin{align}
\ket{\psi}\stackrel{LOCC}{\longrightarrow}\{r_j,\ket{\eta_j}\},\nonumber\\
\mu^{\downarrow}(\ket{\psi})\prec\sum_jr_j\mu^{\downarrow}(\ket{\eta_j}),\label{1mu}
\end{align}
here $\rho=\sum_jr_j\ket{\eta_j}\bra{\eta_j},$ the second equality is due to the results in \cite{jonathan1999minimal}. Assume $\mathcal{E}_{k}$ is a unilocal operation on party $B$, then let $\rho_k=\frac{\mathcal{E}_k(\rho)}{q_k},$ $q_k=\tr\mathcal{E}_k(\rho)$, $\rho_{jk}=\frac{\mathcal{E}_k(\ket{\eta_j})}{t_{jk}}$, $t_{jk}=\tr\mathcal{E}_k(\ket{\eta_j})$. Due to the definition of $\overline{E},$ it is invariant under local unitary operations, it is monotone under the actions 
\begin{align}
\rho\rightarrow \rho\otimes\rho_1,\nonumber\\
\rho\rightarrow \tr_{\mathcal{Q}}\rho,\nonumber
\end{align}
here $\rho_1$ is a state added by one party to its subsystem,  $\mathcal{Q}$ is held by $B$, and $\tr_{\mathcal{Q}}\rho$ is the partial trace on $\mathcal{Q}.$
When $\mathcal{E}_k$ stands for the unilocal von Neumann measurement $\{I\otimes M_k\}$, $\rho_{jk}$ can be pure, and we write $\rho_{jk}=\ket{\xi_{jk}}\bra{\xi_{jk}}$. 
\begin{align}
\mu^{\downarrow}(\ket{\eta_j})\prec \sum_kt_{jk}\mu^{\downarrow}(\rho_{jk}),\label{2mu}
\end{align}
The above equality is due to the Theorem 1 in \cite{jonathan1999minimal}.
\par
 Next let $\ket{\chi_k}$ be a pure state with 
\begin{align}
\mu^{\downarrow}(\ket{\chi_k})=\sum_{j}\frac{r_jt_{jk}}{m_k}\mu^{\downarrow}(\ket{\xi_{jk}}),\label{3mu}
\end{align}
 here $m_k=\sum_{j}r_jt_{jk},$ then 
 \begin{align}
 \ket{\chi_k}\stackrel{LOCC}{\rightarrow}\rho_k,\label{3m}
 \end{align}
 this is due to (\ref{3mu}) and Theorem 1 in \cite{jonathan1999minimal}.\par
   Combing the equality (\ref{1mu}), (\ref{2mu}) and (\ref{3mu}), we have 
 \begin{align}
 \mu^{\downarrow}(\ket{\psi})\prec\sum_k m_k\mu^{\downarrow}(\ket{\chi_k}). \label{4mu}
 \end{align} 
 As $\mathcal{E}_k$ is linear, we have $m_k=q_k,$ that is, 
 \begin{align*}
 \ket{\psi}\stackrel{LOCC}{\rightarrow}\{q_k,\ket{\chi_k}\},
 \end{align*}
 then under the results for pure states, we have 
 \begin{align}
 \overline{E}(\ket{\psi})\ge \sum_k q_k\overline{E}(\ket{\chi_k}), \label{5mu}
\end{align}
then we have
 \begin{align*}
 \overline{E}(\rho)=&\overline{E}(\ket{\psi})\nonumber\\
 \ge &\sum_k q_k\overline{E}(\chi_k)\nonumber\\
 \ge&\sum_k q_k\overline{E}(\rho_k),
 \end{align*}
here  the first inequality is due to the assumption of $\ket{\psi},$ the second inequality is due to (\ref{5mu}), the third inequality is due to (\ref{3m}) and the definition of $\overline{E}$.
	\par 
 Assume $\{q_k, \rho_k\}$ is a decomposition of $\rho_{AB},$ let $\ket{\psi_k}$ be the optimal pure state for $\rho_k$ in terms of the entanglement measure $E,$ and let $\{q_{kl},\ket{\theta_{kl}}\}$ be the corresponding decomposition, by the results in \cite{jonathan1999minimal}, we have 
 \begin{align*}
 \sum_{l}q_{kl}\mu^{\downarrow}(\ket{\theta_{kl}})\succ&\mu^{\downarrow}(\ket{\psi_k}),\\
 \sum_{kl}q_kq_{kl}\mu^{\downarrow}(\ket{\theta_{kl}})\succ&\sum_kq_k\mu^{\downarrow}((\ket{\psi_k}))
 \end{align*} 
 let $\ket{\psi}$ be the pure state such that $\mu^{\downarrow}(\ket{\psi})=\sum_{k}q_k\mu^{\downarrow}((\ket{\psi_k}),$ then we have
 \begin{align}
 \sum_{kl}q_kq_{kl}\mu^{\downarrow}(\ket{\theta_{kl}})\succ\mu^{\downarrow}(\ket{\psi}),
 \end{align}
the due to the result in \cite{jonathan1999minimal}, we have $\ket{\psi}$ can be transformed into $\rho$ under LOCC, then
 \begin{align}
 \overline{E}(\rho)\le E(\ket{\psi})\le \sum_kq_k\overline{E}(\rho_k),
 \end{align}
 here the first inequality is due to the definition of $\overline{E},$ the second inequality is due to the property of $f.$
\end{proof}
\par
From the proof of the above theorem, we may have the following result, it tells us that when we consider the entanglement measure in (1) and $\mathcal{T}$ is LOCC, we could decrease the size of the set of $\rho_{AB}.$
\begin{Theorem}
	Assume that $\rho_{AB}$ is a mixed state in $\mathcal{H}_{AB},$ $E$ is entanglement monotone for pure states, then we have that 
	\begin{align}
	\overline{E}(\rho_{AB})=\inf_{\ket{\psi}_{AB}\in \mathcal{O}(\rho_{AB})}E(\ket{\psi}_{AB}), \label{evt}
	\end{align}
	where we denote $\mathcal{O}(\rho_{AB})$ is the subset of $\mathcal{R}(\rho_{AB})$ with its element $\ket{\psi}$ satisfying $\mu^{\downarrow}(\ket{\psi}_{AB})=\sum_i p_i\mu^{\downarrow}(\ket{\phi_i}_{AB}),$ here $\{p_i,\ket{\phi_i}_{AB}\}$ is a decomposition of $\rho_{AB}.$
\end{Theorem}
\begin{proof}
	By the definition of $\overline{E}(\rho_{AB})$, we have that $\overline{E}(\rho_{AB})\le \inf_{\ket{\psi}_{AB}\in \mathcal{O}(\rho_{AB})}E(\ket{\psi}_{AB}).$ Then we prove the other side of $(\ref{evt}).$ Assume that $\ket{\psi}_{AB}$ is an optimal pure state for the state $\rho_{AB}$ in terms of $\overline{E},$ by the similar analysis in Theorem \ref{th1}, then we have there exists a decomposition $\{p_i,\ket{\phi_i}_{AB}\}$ of $\rho_{AB}$ such that
	\begin{align}
\mu^{\downarrow}(\ket{\psi}_{AB})\prec\sum_i p_i\mu^{\downarrow}(\ket{\phi_i}_{AB}),\label{ode}
\end{align}
the equality (\ref{ode}) is due to the result in \cite{jonathan1999minimal}.
Next if we take $\ket{\psi^{'}}$ with $\mu^{\downarrow}(\ket{\psi^{'}})=\sum_i p_i\mu^{\downarrow}(\ket{\phi_i}_{AB})$, combing with $(\ref{ode}),$ we have that $\mu^{\downarrow}(\ket{\psi}_{AB})\prec\mu^{\downarrow}(\ket{\psi^{'}}_{AB}),$ combing the result in \cite{nielsen1999conditions}, and the definition of $\overline{E},$ we finish the proof. 
\end{proof}
\par Then we make a comparision of $\overline{E}$ with $E_f$ for a mixed state $\rho_{AB}$.
\begin{Theorem}
	Assume that $E$ is an entanglement measure for pure states in $\mathcal{H}_{AB},$ and  $\overline{E}$ is an entanglement measure for a mixed state defined as (\ref{Ev}), then $\overline{E}$ is convex if and only if $\overline{E}=E_f$
\end{Theorem}
\begin{proof}
	As when $\rho_{AB}=\ket{\psi}_{AB}\bra{\psi}$ is a pure state in $\mathcal{H}_{AB},$ $E_f(\ket{\psi}_{AB})=\overline{E}(\ket{\psi}),$ then by the result in \cite{gour2020optimal}, we have that $\overline{E}(\rho_{AB})\ge E_f(\rho_{AB}).$ On the other hand, from the definition of the $E_f,$ $E_f$ is convex. Assume that $\{q_j,\ket{\theta_j}\}$ is the optimal decomposition of $\rho_{AB}$ in terms of $E_f,$ then we have that 
	\begin{align}
	\overline{E}(\rho_{AB})\le &\sum_jq_jE(\ket{\theta_j})\nonumber\\
	=&E_f(\rho_{AB}),
	\end{align}
	the inequality is due to the convexity of $\overline{E},$
	then we finish the proof.
\end{proof}\par
Then we present a condition when $\overline{E} $ is convex.
	\begin{Theorem}
		Assume $\rho$ is a bipartite entangled state, and $\rho$ can be written as $\rho=p_1\sigma_1\oplus p_2\sigma_2,$ here $p_1\sigma_1\oplus p_2\sigma_2$ means that $\supp(\sigma_1)\cap\supp(\sigma_2)=\emptyset,$ $i.e.$ 
	$\rho=\begin{pmatrix} p_1\sigma_1 &  \\  & p_2\sigma_2 \end{pmatrix}.$ And let $E$ be an entanglement measure for a bipartite pure state $\ket{\phi}\in \mathcal{H}_{AB}$, $E(\ket{\phi})=f(\tr_B\ket{\phi}\bra{\phi})$, if $f$ is convex, then $\overline{E}(\rho)\le p_1\overline{E}(\sigma_1)+p_2\overline{E}(\sigma_2).$ 
	\end{Theorem}
\begin{proof}
	Assume that $\ket{\phi_i}$ is the optimal pure state for $\sigma_i$ in terms of the entanglemennt measure $\overline{E},$ $i=1,2,$ then there exists a decomposition $\{q_k,\ket{\varphi_k^i}\}$ of $\sigma_i$ such that $\mu^{\downarrow}(\ket{\phi_i})\prec\sum_k q_k\mu^{\downarrow}(\ket{\varphi_k^i})$, $\sum_{i=1}^{2}p_i\mu^{\downarrow}(\ket{\phi_i})\prec\sum_{i=1}^{2}\sum_k p_iq_k\mu^{\downarrow}(\ket{\varphi_k^i})$. Next by the Ky-Fan's maximum principle \cite{bhatia2013matrix}, we have that $f_k$ is a concave function, and combing $\supp(\sigma_1)\cap\supp(\sigma_2)=\emptyset,$ we have $E_k(\sqrt{p_1}\ket{\phi_1}+\sqrt{p_2}\ket{\phi_2})\ge p_1E_k(\ket{\phi_1})+p_2E_k(\ket{\phi_2})$, $k=1,2,3,\cdots,d,$ then $\mu^{\downarrow}(\sqrt{p_1}\ket{\phi_1}+\sqrt{p_2}\ket{\phi_2})\prec\sum_{i=1}^{2}\sum_k p_iq_k\mu^{\downarrow}(\ket{\varphi_k^i}),$ $\sqrt{p_1}\ket{\phi_1}+\sqrt{p_2}\ket{\phi_2}\stackrel{LOCC}{\longrightarrow}\rho$. Next as we assume $f$ is convex, then we have 
	\begin{align}
	&p_1\overline{E}(\sigma_1)+p_2\overline{E}(\sigma_2)\nonumber\\
	\ge &E(\sqrt{p_1}\ket{\phi_1}+\sqrt{p_2}\ket{\phi_2})\nonumber\\
	\ge &\overline{E}(\rho),
	\end{align}  
	the first inequality is due to the convexity of the function $f,$ the second inequality is due to the definition of $\overline{E}$ in (\ref{Ev}).
\end{proof}

\par An important property of entanglement measure is additivity, it means that $\forall\sigma\in\mathcal{H}_{AB},$ $E(\sigma^{\otimes n})=n E(\sigma),$ if $E(\sigma^{\otimes n})\le n E(\sigma),$ we say  $E$ is subadditivity. Unfortunely, this property is not always valid for many prominent entanglement measures, such as, entanglement of formation \cite{bennett1996mixed}, robustness of entanglement of entanglement \cite{vidal1999robustness}, relative entropy of entanglement \cite{vedral1998entanglement,vedral2002role}. Moreover, the relative entropy of entanglement is additivity for pure states, while it is subadditivity for mixed states. Here we present a condition when $\overline{E}$ is weak subadditivity.
\begin{Theorem}
	Assume $E$ is an entanglement  measure for a pure state $\ket{\psi}_{AB}$ in $\mathcal{H}_{AB}$, and $E$ is subadditive for pure states. When $\rho$ is a bipartite mixed state on $\mathcal{H}_{AB},$
	\begin{align}
	\overline{E}(\rho^{\otimes n})\le n \overline{E}(\rho),
\end{align}
\end{Theorem}
\begin{proof}
	Assume $\ket{\psi}$ is the optimal pure state for a mixed state $\rho$ in terms of $\overline{E},$ then $\rho^{\otimes n}$ can be transformed into $\ket{\psi}^{\otimes n}$ by LOCC. Due to the definition of $\overline{E},$ we have 
	\begin{align}
	\overline{E}(\rho^{\otimes n}) \le& \overline{E}(\ket{\psi}^{\otimes n})\nonumber\\
	\le&n\overline{E}(\ket{\psi})\nonumber\\=&n\overline{E}(\rho).
	\end{align} 
	The first inequality is due to the definition of $\overline{E}$ in (\ref{Ev}), the second inequality is due to the subadditivity of $E$ for pure states, the first equality is due to the assumption that $\ket{\psi}$ is the optimal for $\rho$ in terms of $\overline{E}$.
\end{proof}
\section{An operational interpretation of $\overline{E}$}
\indent In this section, we will present the interpretation of the smoothed quantum entanglement cost built under the method here. Entanglement cost (entanglement distillation) means the optimal rate $r$ from (to)  the two qubit maximally entangled state
\begin{align}
\Psi_r=\frac{1}{r}\sum_{i=0}^{r-1}\sum_{j=0}^{r-1}\ket{ii}\bra{jj},
\end{align}
\begin{definition}\cite{buscemi2011entanglement}
	Assume $\rho\in \mathcal{D}(\mathcal{H}_{AB})$, its one-shot entanglement cost is defined as
	\begin{align}
	E_{c,1}(\rho)=\log\inf_{\Lambda\in LOCC}\{r|\Lambda(\Psi_r)=\rho\},
	\end{align} 
	 the smoothed entanglement cost is defined as
	\begin{align}
	E_{c,1}^{\epsilon}(\rho)=\inf_{\overline{\rho}\in B_{\epsilon}(\rho)} E_{c,1}(\overline{\rho}),
	\end{align}
	here $B_{\epsilon}(\rho)=\{\overline{\rho}|\frac{1}{2}||\rho-\overline{\rho}||\le \epsilon\}$,  $||A||=\tr\sqrt{A^{\dagger}A}$.
		Its one-shot smoothed entanglement distillation is defined as
	\begin{align}
	E^{\epsilon}_{d,1}(\rho)=\log\sup_{\Lambda\in LOCC}\{r|\frac{1}{2}||\Lambda(\rho)-\Psi_r||\le \epsilon\}
	\end{align}
	\end{definition}
Next we recall the definition of entanglement formation and entanglement cost.
\begin{definition}\cite{bennett1996mixed}
	Assume $\ket{\psi}_{AB}\in \mathcal{H}_{AB},$ its entanglement of formation is defined as
	\begin{align}
E_f(\ket{\psi})=S(\tr_B\ket{\psi}_{AB}\bra{\psi}),
	\end{align}
	here $S(\rho)=-\tr\rho\log\rho.$\par
When $\rho\in\mathcal{D}(\mathcal{H}_{AB})$ is a mixed state, its entanglement of formation is defined by the convex roof extended method,
\begin{align}
E_f(\rho)=\min_{\{p_i,\ket{\psi}\}}\sum_ip_iE_f(\ket{\psi}),
\end{align}
where the minimization takes over all the decompositions $\{p_i,\ket{\psi}\}$ of $\rho$ such that $\rho=\sum_ip_i\ket{\psi}\bra{\psi}.$\par
Assume $\rho\in \mathcal{D}(\mathcal{H}_{AB})$, its entanglement cost is 
\begin{align}
E_c(\rho)=\log\inf\{r|\lim\limits_{n\rightarrow\infty}\inf_{\Lambda\in LOCC}\tr|\rho^{\otimes n}-\Lambda(\Psi_r)|=0\}\label{ec}
\end{align}
\end{definition}\par
Next we present the similar smoothed entanglement measure defined in (\ref{Ev}).
\begin{definition}\cite{regula2019one}
	Assume $\rho_{AB}\in \mathcal{D}(\mathcal{H}_{AB})$, $\overline{E}$ is an entanglement measure defined in $(\ref{Ev}),$ the smoothed extension of $\overline{E},$ $\overline{E}^{\epsilon}$ is defined as
	\begin{align}
	\overline{E}^{\epsilon}(\rho)=&\inf_{\overline{\rho}\in B_{\epsilon}(\rho)}\overline{E}(\overline{\rho})\nonumber\\
	=&\inf_{\overline{\rho}\in B_{\epsilon}(\rho)}\inf\{E(\ket{\psi})|\Lambda(\ket{\psi})=\overline{\rho},\Lambda\in LOCC\}.
	\end{align}
	Here we restrict $\mathcal{T}$ defined in $(\ref{Ev})$ to be LOCC, and the second inf takes over all the pure states $\ket{\psi}$ such that $\Lambda(\ket{\psi})=\overline{\rho},$ $\Lambda\in LOCC.$
\end{definition}
\par When we take $E$ for pure states as entanglement cost, we have the following theroem.
\begin{Theorem}
	Assume $\rho_{AB}\in\mathcal{D}(\mathcal{H}_{AB})$, $\epsilon\in(0,1)$, then $\overline{E_{c,1}^{\epsilon}}(\rho)=E^{\epsilon}_{c,1}(\rho).$
\end{Theorem}
\begin{proof}
	Assume $\ket{\psi_{\epsilon}}$ is the optimal pure state for $\rho$ in terms of $\overline{E_{c,1}^{\epsilon}}$, and let 
	\begin{align}
	\overline{E_{c,1}^{\epsilon}}(\rho)=E_{c,1}(\ket{\psi_{\epsilon}})=r_{\epsilon},
	\end{align}
	next as when $\Lambda$ is a quantum channel, $||\Lambda(\rho)||_1\le ||\rho||_1,$ then let $\Omega(\ket{\psi_{\epsilon}})=\rho_{\epsilon},$ here $\rho_{\epsilon}\in B_{\epsilon}(\rho)$ we have 
	\begin{align}
	&||\Omega^{\otimes n}\circ\Lambda(\Psi_{r_{\epsilon}}^{\otimes n})-\Omega^{\otimes n}(\ket{\psi_{\epsilon}}^{\otimes n})||\nonumber\\\le& ||\Lambda(\Psi_{r_{\epsilon}}^{\otimes n})-(\ket{\psi_{\epsilon}}^{\otimes n})|| \le\epsilon\label{lhs}
	\end{align}
	that is, $\overline{E_{c,1}^{\epsilon}}(\rho)\ge E_{c,1}^{\epsilon}(\rho)$.\par
	Next we prove the other side. Assume $\rho_{\epsilon}$ is the optimal of $\rho$ in terms of $E_{c,1}^{\epsilon}(\rho)$ and $E_{c,1}^{\epsilon}(\rho)=r_{\epsilon}$. Let $\ket{\phi_{\epsilon}}$ be the optimal for $\rho_{\epsilon}$ in terms of $\overline{E_{c,1}}.$ Then we would show that $E_{c,1}(\ket{\phi_{\epsilon}})=r_{\epsilon}$.\par 
	First $E_{c,1}(\ket{\phi_{\epsilon}})> r_{\epsilon}$ is impossible, as by the definition of $E_{c,1},$ $\Psi_{ r_{\epsilon}}$ can be the optimal pure state for $\rho_{\epsilon}$ in terms of $E_{c,1}.$ Next if $E_{c,1}(\ket{\phi_{\epsilon}})< r_{\epsilon},$ then by the similar thought of $(\ref{lhs})$, we see it is impossible, that is, 
	\begin{align}
	E_{c,1}(\ket{\phi_{\epsilon}})=\overline{E_{c}}(\rho_{\epsilon})=r_{\epsilon},
	\end{align}
	next by the definition of $\overline{E_{c}^{\epsilon}},$ we have $\overline{E_{c}^{\epsilon}}(\rho)\le E_{c,1}^{\epsilon}(\rho)$. Then we finish the proof.
\end{proof}\par
From the proof of the above theorem, we donot use the property of LOCC, that is, when the set $\mathcal{T}$ of operations in $\overline{E_{c}^{\epsilon}}$  is in line with the set of operations in $E_{c,1}^{\epsilon},$ $i.e.$ $\mathcal{T}=SEP, PPT$ or $SEPP,$ the above result is also valid.\par

\begin{Theorem}
	Assume $\rho\in\mathcal{D}(\mathcal{H}_{AB})$, then we have 
	\begin{align}
	\lim\limits_{\epsilon\rightarrow0}\lim\limits_{n\rightarrow \infty}\frac{\overline{E_f^{\epsilon}}(\rho^{\otimes n})}{n}=E_c(\rho).
	\end{align}
\end{Theorem}
\begin{proof}
From the definition of entanglement cost in (\ref{ec}), we have that $\forall\epsilon,$ when $n\rightarrow\infty,$ there exists a $\Lambda\in LOCC$ such that $\Lambda({\Psi_r^{\epsilon}}^{\otimes n})=\rho_{n,\epsilon},$ $\rho_{n,\epsilon}\in B_{\epsilon}(\rho^{\otimes n})$, that is
\begin{align}
E_c({\Psi_r^{\epsilon}}^{\otimes n})\ge \overline{E_c}(\rho_{n,\epsilon})\ge \overline{E_c^{\epsilon}}(\rho^{\otimes n}),
\end{align}
the second inequality is due to the definition of the smoothed entanglement cost. Next denote $\sigma_{\epsilon}$ as the state such that $\overline{E_c^{\epsilon}}(\rho^{\otimes n})=\overline{E_c}(\sigma_{\epsilon})$,  then as when $\ket{\psi}_{AB}$ is a pure state, $E_c(\ket{\psi}_{AB})=E_f(\ket{\psi}_{AB})=-\tr\rho_B\log\rho_B$, and the Theorem 1 in \cite{gour2020optimal}, we have 
\begin{align}
\overline{E_c}(\sigma_{\epsilon})\ge E_f(\sigma_{\epsilon}),
\end{align}
In \cite{winter2016tight}, the author showed that when $\rho,\sigma\in\mathcal{D}(\mathcal{H}_{AB}),$ $\frac{1}{2}||\rho-\sigma||\le \epsilon,$
\begin{align*}
|E_f(\rho)-E_f(\sigma)|\le \delta\log d+(1+\delta)h(\frac{\delta}{\delta+1}),\label{aw}
\end{align*}
here $d$ is the dimension of the smaller of the two system. Without loss of generality, we assume $\dim\mathcal{H}_A=\dim\mathcal{H}_B=d,$ $\delta=\sqrt{\epsilon(2-\epsilon)},$ $h(\epsilon)=-\epsilon\log\epsilon-(1-\epsilon)\log(1-\epsilon).$
\begin{align*}
\overline{E^{\epsilon}_{c}}(\rho^{\otimes n})\ge& E_f(\sigma_{\epsilon})\nonumber\\
\ge& E_f(\rho^{\otimes n}) -n\delta\log d+(1+\delta)h(\frac{\delta}{1+\delta}),
\end{align*}
that is,
\begin{align}
\frac{\overline{E^{\epsilon}_{c}}(\rho^{\otimes n})}{n}\ge \frac{E_f(\rho^{\otimes n})}{n}-\delta\log d +\frac{(1+\delta)h(\frac{\delta}{1+\delta})}{n}.
\end{align}
As $\lim\limits_{\epsilon\rightarrow 0} h(\epsilon)=0$ and $\lim\limits_{\epsilon\rightarrow0} \delta=0,$
then we have 
\begin{align}
E_c(\rho)=&\lim\limits_{\epsilon\rightarrow0}\lim\limits_{n\rightarrow\infty}\frac{{E_c(\Psi_r^{\epsilon}}^{\otimes n})}{n}\nonumber\\
\ge&\lim\limits_{\epsilon\rightarrow0}\lim\limits_{n\rightarrow\infty}\frac{\overline{E^{\epsilon}_{c}}(\rho^{\otimes n})}{n}\nonumber\\
\ge& \lim\limits_{n\rightarrow\infty}\frac{E_f(\rho^{\otimes n})}{n}\nonumber\\
=&E_c(\rho).
\end{align}
The last equality is due to the result in \cite{hayden2001asymptotic},
then we finish the proof.
\end{proof}
\section{Applications}
\par In the following, we presented that a difference between the quantity defined in (\ref{Ev}) between LOCC and SEPP. Then we present that the entanglement generated from the geometric entanglement measure for pure states by the convex roof extended method and the method proposed here are equal. At last, we present that for the $2\otimes 2$ system, the entanglement generated from concurrence for pure states by the convex roof extended method and the method proposed here are equal, and we present that the entanglement measure generated by concurrence for pure states by the method proposed here is monogamous under the definition of monnogamy proposed in \cite{gour2018monogamy}.
\subsection{An Example on an entanglement measure under SEPP}
\begin{definition}\cite{terhal2000schmidt}
	Assume $\ket{\psi}_{AB}$ is a pure state,  its Schimidt number 
	\begin{align*}
	Sch(\ket{\psi_{AB}})=Rank(\rho_A),
	\end{align*} here $\rho_A=\tr_B\ket{\psi}_{AB}\bra{\psi}$. \par
	When $\rho_{AB}$ is a mixed state, then its Schimidt number $Sch(\rho)$ is $k$, if (i) there exists a decomposition of $\{p_i,\ket{\psi_i}\}$ such that the Schimidt number of all the pure states $\ket{\psi_i}$ are at most $k,$ (ii) for any decomposition $\{p_i,\ket{\psi_i}\}$ of $\rho_{AB},$ there exists at least one pure state $\ket{\psi_j}$ in the set $\{\ket{\psi_i}\}$ with its Schmidt number at least $k$. 
\end{definition}
\par  In \cite{gour2020optimal}, the authors showed that when the entanglement measure $E$ is the Schmidt number, $\mathcal{T}=LOCC,$ $\overline{E}(\rho_{AB})=E(\rho_{AB}).$ In \cite{chitambar2020entanglement}, the authors presented the following interesting result.
\begin{Lemma}\label{net}\cite{chitambar2020entanglement}
	For every biparitite state $\rho$ and any positive interger $k$, there exists a SEPP operation $\Lambda$ such that $\Lambda(\ket{\psi_k})=\rho_{AB}$ if and only if $R(\rho)\le R(\ket{\psi})$, here $\ket{\psi_k}=\frac{1}{\sqrt{k}}\sum_i\ket{ii}$ is a maximally entangled state, and $R(\rho)$ is its robustness of entanglement which is defined as follows
	\begin{align*}
	R(\rho)=\min_{\sigma\in \mathcal{S}(\mathcal{H}_{AB})}min_s\{s|\rho+s\sigma/(1+s)\in \mathcal{S}(\mathcal{H}_{AB})\}.
	\end{align*}
\end{Lemma}\par
Then we present that when $\mathcal{T}\in SEPP,$ the entanglement measure built from the method here can be increased.
\begin{example}
	Assume that $\ket{\psi}=\frac{1}{\sqrt{3}}(\ket{00}+\ket{11}+\ket{22}),$ let
	\begin{align*}
	\rho=&\frac{1}{4}\ket{\phi_1}\bra{\phi_1}+\frac{3}{4}\ket{\phi_2}\bra{\phi_2},\\ \ket{\phi_1}=&\frac{1}{2}\ket{00}+\frac{1}{6}\ket{11}+\frac{1}{6}\ket{22}+\frac{5}{6}\ket{33},\\ \ket{\phi_2}=&\frac{1}{2}\ket{00}+\frac{1}{8}\ket{11}+\frac{1}{8}\ket{22}+\frac{\sqrt{46}}{8}\ket{33},
	\end{align*} in \cite{vidal1999robustness}, the authors showed that when $\ket{\phi}=\sum_i\sqrt{\lambda_i}\ket{ii}$, $R(\ket{\phi})=(\sum_i\sqrt{\lambda_i})^2-1$, then
	\begin{align*}
	R(\ket{\psi})=&2,\\
	R(\ket{\phi_1})=&1.7778,\hspace{2mm} R(\ket{\phi_2})=1.5529.
	\end{align*}\par
	Next as $R$ is convex, we have $R(\rho)\le 1.7778<2,$ due to the Lemma \ref{net}, we have that there exists a SEPP $\Lambda$ such that $\Lambda(\ket{\psi})=\rho.$ However, it is clear to see that $Sch(\ket{\psi})=3,$ $Sch(\rho)=4,$ that is, when $\mathcal{T}$ stands for the LOCC in (\ref{Ev}), $\overline{E}(\rho)=4,$ when $\mathcal{T}$ is SEPP in (\ref{Ev}),  $\overline{E}(\rho)\le 3.$
\end{example}
\subsection{The extension of geometric entanglement measure}
Here we first discuss the connection between the extension of geometric entanglement measure by the method of (\ref{Ev}) and the original definition defined in $\r$ \cite{wei2003geometric}. The latter is defined as the maximum overlap between $\r$ and any fully product states $\ket{a_1,...,a_n}$. That is,
\begin{eqnarray}
G(\r):=1-\max_{\ket{\ps}=\ket{a_1,...,a_n}}	
\bra{\ps}\r\ket{\ps}.\label{gem}
\end{eqnarray}
The GME is a fundamental multipartite entanglement measure in the past decades \cite{chen2010computation, chen2011Connections, zhu2010additivity}. The GME can quantify the entanglement of experimentally realizable states, GHZ states, W states and graph states for one-way quantum computing \cite{briegel2001persistent}, topological quantum computing \cite{lu2009demonstrating}, and six-photon Dicke states \cite{wieczorek2009experimental}, respectively. \\
\indent Assume $\ket{\psi}=\sum_i\sqrt{\lambda_i}\ket{ii}\in \mathcal{H}_{AB},$ here we can assume $\lambda_j\ge \lambda_{j+1},$ then from the definition of $(\ref{gem})$, we have
\begin{align}
G(\ket{\psi})=&1-\max_{\phi}|\bra{\phi}\psi\rangle|^2\nonumber\\
                   =&1-\lambda_0.
\end{align}\par
The extension of geometric entanglement measure is defined as
\begin{align}
\overline{G}(\rho)=\inf_{\psi\in R(\rho)}G(\ket{\psi}), \label{gmep}
\end{align}
here $R(\rho)$ is the set of pure states that can be tansformed into $\rho$ through LOCC.  Next we prove that 
\begin{Theorem}
	Assume $\rho\in \mathcal{D}(\mathcal{H}_{AB}),$ then
		\begin{align}
		\overline{G}(\rho)=G_f(\rho),\label{cgme},
		\end{align}
		here $G_f$ for mixed states is built by the convex roof extended method defined in (\ref{cre}).
	\end{Theorem}
\begin{proof}
	As when $\ket{\psi}\in \mathcal{H}_{AB},$ $G_f(\ket{\psi})=\overline{G}(\ket{\psi}),$ by the Theorem 1 in \cite{gour2020optimal}, we have 
	\begin{align}
	\overline{G}(\rho)\ge G_f(\rho). \label{lgme}
	\end{align}
	 Next assume $\{p_i,\ket{\phi_i}\}$ is the optimal decomposition of $\rho$ in terms of $G_f$, assume $\ket{\varphi}$ is the pure state with $\mu^{\downarrow}(\varphi)=\sum_ip_i\mu^{\downarrow}(\phi_i)$, by the main result in \cite{jonathan1999minimal},  $\varphi\stackrel{LOCC}{\longrightarrow}\rho,$ by the definition of $\overline{G},$ we have
	\begin{align}
    \overline{G}(\rho)\le G_f(\rho). \label{rgme}
	\end{align}
	Combing with (\ref{lgme}) and (\ref{rgme}), we finish the proof.
\end{proof}
\subsection{Some results on $2\otimes 2$ states}
\indent Here we first recall the definition of concurrence for bipartite quantum states. Assume $\ket{\psi}_{AB}\in \mathcal{H}_{AB},$ 
its concurrence \cite{bennett1996mixed} is defined as
\begin{align}
C(\ket{\psi}_{AB})=\sqrt{2(1-\tr\rho_A^2)},
\end{align}
here $\rho_A=\tr_B\ket{\psi}_{AB}\bra{\psi}.$ When $\rho_{AB}$ is a bipartite mixed state, its concurrence is defined as
\begin{align}
C(\rho_{AB})=\min_{\{p_i,\ket{\phi_i}\}}\sum_ip_iC(\ket{\phi_i}),
\end{align} 
where the minimum takes over all the decompositions $\{p_i,\ket{\phi_i}\}$ of $\rho_{AB}$ such that $\rho_{AB}=\sum_ip_i\ket{\phi_i}\bra{\phi_i}.$\par
 Moreover, when $\ket{\psi}_{AB}$ is bipartite qubit pure state,
By the Schimidt decomposition, $\ket{\psi}_{AB}$ can be written as
\begin{align}
\ket{\psi}_{AB}=\sqrt{\lambda_0}\ket{00}+\sqrt{\lambda_1}\ket{11},
\end{align} here we assume $1\ge \lambda_0\ge \lambda_1\ge 0,$ $\lambda_0+\lambda_1=1.$ And by computation, its concurrence is $C^2(\ket{\psi}_{AB})=4\lambda_0(1-\lambda_0),$ that is, 
\begin{align}
\lambda_0=\frac{1+\sqrt{1-C^2}}{2},\hspace{3mm}\lambda_1=\frac{1-\sqrt{1-C^2}}{2}.\label{lambda}
\end{align}
\par
Next we present the relation between $C$ and $\overline{C}$ of a bipartite state $\rho_{AB}.$
\begin{Theorem}\label{c}
	Assume $\rho_{AB}\in\mathcal{D}(\mathcal{H}_2\otimes\mathcal{H}_2),$ then we have 
	\begin{align}
	\overline{C}(\rho_{AB})=C(\rho_{AB}).
	\end{align}
\end{Theorem}
\begin{proof}
	Assume $\{p_i,\ket{\phi_i}\}$ is the optimal decomposition of $\rho_{AB}$ in terms of $C$, that is, for any decomposition $\{q_k,\ket{\varphi_k}\}$ of $\rho_{AB},$
	\begin{align*}
	\sum_ip_iC(\ket{\phi_i})\le \sum_kq_kC(\ket{\varphi_k}).
	\end{align*}
	Let $\ket{\chi}$ be a pure state with $\mu^{\downarrow}(\ket{\chi})=\sum_i p_i\mu^{\downarrow}(\phi_i)$, then by the theorem 1 in \cite{jonathan1999minimal},  we have $\ket{\chi}$ can be transformed into $\ket{\phi_i}$ with probability $p_i,$ that is, $\ket{\chi}\stackrel{LOCC}{\longrightarrow}\rho_{AB}.$\par
Next in \cite{wootters1998entanglement}, the authors showed that for a two-qubit state $\rho_{AB}$, there exists a decomposition $\{r_l,\ket{\omega_l}\}$ of $\rho_{AB}$ such that
\begin{align}
C(\rho_{AB})=&\min_{\{r_l,\ket{\omega_l}\}} \sum_lr_lC(\ket{\omega_l}),\nonumber\\
C(\ket{\omega_l})=&C(\rho_{AB}), \hspace{3mm}\forall l.
\end{align}
Then we have
\begin{align}
\sum_lr_l\sqrt{1-C^2(\ket{\omega_l})}=&\sqrt{1-C^2(\rho_{AB})}\nonumber\\
\ge&\sqrt{1-(\sum_kq_kC(\ket{\varphi_k})^2}\nonumber\\
\ge&\sum_k q_k\sqrt{1-C^2(\ket{\varphi_k})},
\end{align}
here we denote that $\{q_k,\ket{\varphi_k}\}$ is an arbitrary decomposition of $\rho_{AB}$. The first inequality is due to the definition of concurrence, the second inequality is due to the Cauchy-Schwarz inequality. The by the equality (\ref{lambda}), we have that 
\begin{align}
\sum_lr_l\mu^{\downarrow}(\ket{\omega_l})\succ\sum_k q_k\mu^{\downarrow}(\ket{\varphi_k}).
\end{align}
Next we prove that the state $\ket{\upsilon}$ with $\mu^{\downarrow}(\ket{\upsilon})=\sum_lr_l\mu^{\downarrow}(\ket{\omega_l})$ is the optimal for $\rho_{AB}$ in terms of $\overline{C}.$ First if $\ket{\xi}\in \mathcal{R}(\rho_{AB}),$ then there exists a decomposition $\{m_t,\ket{\phi_t}\}$
\begin{align}
\mu^{\downarrow}(\ket{\xi})\prec \sum_t m_t\mu^{\downarrow}(\ket{\phi_t})\prec\sum_l r_l\mu^{\downarrow}(\omega_l),
\end{align} 
on the other hand, as $C(\ket{\omega_l})=C(\rho)$, $\forall l,$ and $\lambda_0=\frac{1+\sqrt{1-C^2}}{2},$ $\lambda_1=\frac{1-\sqrt{1-C^2}}{2}$ then we have that 
\begin{align}
\mu^{\downarrow}(\ket{\upsilon})=\mu^{\downarrow}(\ket{\omega_l})\succ\mu^{\downarrow}(\ket{\xi}),
\end{align}
that is, $\ket{\xi}\stackrel{LOCC}{\longrightarrow}\ket{\upsilon}.$ Due to the definition of $\overline{C},$ we have that $\ket{\upsilon}$ is the optimal pure state for $\rho_{AB}$ in terms of $E$. Last, by the above analysis and the definition of $\overline{C}(\rho_{AB}),$ we have that 
\begin{align}
C(\rho_{AB})=\overline{C}(\rho_{AB}).
\end{align}
Then we finish the proof.
\end{proof}
\par Here we remark that from the proof of the above theorem, other entanglement measures in terms of the method proposed here for the states in $2\otimes 2$ systems can be obtained, such as Tsallis-$q$ entanglement measure \cite{rossignoli2002generalized} and R$\acute{e}$nyi-$\alpha$ entanglement measure \cite{horodecki2001separability} when $q$ \cite{luo2016general} and $\alpha$ \cite{song2016general} are in some regions. \par
   Monogamy of entanglement (MoE) is a fundamental property that can distinguish entanglement from classical correlations. Mathematically, MoE means that it can be characterized as in terms of an entanglement measure $\mathcal{E}$ for a tripartite system $A,B$ and $C$,
   \begin{align*}
   \mathcal{E}_{A|BC}\ge \mathcal{E}_{AB}+\mathcal{E}_{AC},
   \end{align*} 
   here $\mathcal{E}_{AB}$ denotes the entanglement $AB$ in terms of $\mathcal{E}.$ Although many entanglement measures satisfy the above inequality for multi-qubit systems \cite{coffman2000distributed,christandl2004squashed,san2010tsallis,bai2014general,luo2016general}, the above inequality is not valid in general in terms of almost all entanglement measures for multipartite higher dimensional systems \cite{ou2007violation,bai2014general},  it seems only one known entanglement measure, the squashed entanglement, is monogamous for arbitrary dimensional systems \cite{christandl2004squashed}. \\
   \indent Recently, a generalized monogamy relation for an entanglement measure $\mathcal{E}$ was proposed in \cite{gour2018monogamy}.  There the authors defined that an entanglement measure $\mathcal{E}$ is monogamous for a tripartite system $A,B$ and $C$ if for any $\rho_{ABC}\in \mathcal{H}_{A}\otimes\mathcal{H}_B\otimes\mathcal{H}_C,$
   \begin{align}
   \mathcal{E}_{A|BC}=\mathcal{E}_{AB}\Longrightarrow \mathcal{E}_{AC}=0.
   \end{align}\par
   Moreover, the authors showed a class of entanglement measures satisfies the above relation for tripartite systems \cite{guo2019monogamy}.  Next we present a corollary due to the Theorem \ref{c}.
   \begin{Corollary}
   	Let $\rho_{ABC}\in \mathcal{D}(\mathcal{H}_2\otimes\mathcal{H}_2\otimes\mathcal{H}_d),$ then if $\overline{C}(\rho_{A|BC})=\overline{C}(\rho_{AB}),$ then $\overline{C}(\rho_{AC})=0.$
   \end{Corollary} 
\begin{proof}
	Due to the Theorem \ref{c} and the assumption, we have 
	\begin{align}
	\overline{C}(\rho_{A|BC})=\overline{C}(\rho_{AB})=C(\rho_{AB}),\label{cc1}
	\end{align}
	next by the Theorem 1 in \cite{guo2019monogamy}, 
	\begin{align}
	\overline{C}(\rho_{A|BC})\ge C(\rho_{A|BC}),\label{cc2}
	\end{align}
	then combing the (\ref{cc1}) and (\ref{cc2}), we have that 
	\begin{align}
	\overline{C}(\rho_{A|BC})\ge& C(\rho_{A|BC})\nonumber\\
	\ge& C(\rho_{AB})\nonumber\\
	=&\overline{C}(\rho_{AB}),
	\end{align}
	that is, $C(\rho_{A|BC})=C(\rho_{AB}).$ As in \cite{guo2019monogamy}, the authors presented that $C$ is monogamous, then $C(\rho_{AC})=0,$ that is, $\rho_{AC}$ is separable. On the other hand, a seprable pure state can be transformed into a separable mixed state through LOCC, then we have 
	\begin{align}
	\overline{C}(\rho_{AC})=0.
	\end{align}
\end{proof}
\section{Conclusion}
\indent In the paper, we have presented an approach to bulid an entanglement measure for mixed states based on the measure for pure states. First we have presented when the entanglement measure is entanglement monotone, convex and subaddivity, we also have considered the relationship between the entanglement measure built by the convex roof extended method and the method proposed here. Then we present the relation between the smoothed one-shot entanglement cost and the entanglement cost built from the method proposed here, which may present an operational interpretation of the latter entanglement measure. At last, we have presented some applications, first we have presented an example, it told us a difference between the measure bulit from the method here under SEPP and LOCC, then we have presented the equality between the entanglement measure generated from the geometric entanglement measure for pure states under the convex roof extended method and the method here, we also have presented that for $2\otimes 2$ systems, the entanglement measure generated from the concurrence for pure states under the convex roof extended method and the method here are equal, which can show the entanglement measure bulit from our method is monogamous for $2\otimes2\otimes d$.
\bibliographystyle{IEEEtran}
\bibliography{ref}

\begin{thebibliography}{10}
\providecommand{\url}[1]{#1}
\csname url@rmstyle\endcsname
\providecommand{\newblock}{\relax}
\providecommand{\bibinfo}[2]{#2}
\providecommand\BIBentrySTDinterwordspacing{\spaceskip=0pt\relax}
\providecommand\BIBentryALTinterwordstretchfactor{4}
\providecommand\BIBentryALTinterwordspacing{\spaceskip=\fontdimen2\font plus
\BIBentryALTinterwordstretchfactor\fontdimen3\font minus
  \fontdimen4\font\relax}
\providecommand\BIBforeignlanguage[2]{{%
\expandafter\ifx\csname l@#1\endcsname\relax
\typeout{** WARNING: IEEEtran.bst: No hyphenation pattern has been}%
\typeout{** loaded for the language `#1'. Using the pattern for}%
\typeout{** the default language instead.}%
\else
\language=\csname l@#1\endcsname
\fi
#2}}

\bibitem{horodecki2009quantum}
R.~Horodecki, P.~Horodecki, M.~Horodecki, and K.~Horodecki, ``Quantum
  entanglement,'' \emph{Reviews of modern physics}, vol.~81, no.~2, p. 865,
  2009.

\bibitem{plenio2014introduction}
M.~B. Plenio and S.~S. Virmani, ``An introduction to entanglement theory,'' in
  \emph{Quantum Information and Coherence}.\hskip 1em plus 0.5em minus
  0.4em\relax Springer, 2014, pp. 173--209.

\bibitem{ekert1991quantum}
A.~K. Ekert, ``Quantum cryptography based on bell’s theorem,'' \emph{Physical
  review letters}, vol.~67, no.~6, p. 661, 1991.

\bibitem{bennett1993teleporting}
C.~H. Bennett, G.~Brassard, C.~Cr{\'e}peau, R.~Jozsa, A.~Peres, and W.~K.
  Wootters, ``Teleporting an unknown quantum state via dual classical and
  einstein-podolsky-rosen channels,'' \emph{Physical review letters}, vol.~70,
  no.~13, p. 1895, 1993.

\bibitem{bennett1992communication}
C.~H. Bennett and S.~J. Wiesner, ``Communication via one-and two-particle
  operators on einstein-podolsky-rosen states,'' \emph{Physical review
  letters}, vol.~69, no.~20, p. 2881, 1992.

\bibitem{bennett1996mixed}
C.~H. Bennett, D.~P. DiVincenzo, J.~A. Smolin, and W.~K. Wootters,
  ``Mixed-state entanglement and quantum error correction,'' \emph{Physical
  Review A}, vol.~54, no.~5, p. 3824, 1996.

\bibitem{vedral1997quantifying}
V.~Vedral, M.~B. Plenio, M.~A. Rippin, and P.~L. Knight, ``Quantifying
  entanglement,'' \emph{Physical Review Letters}, vol.~78, no.~12, p. 2275,
  1997.

\bibitem{vidal2000entanglement}
G.~Vidal, ``Entanglement monotones,'' \emph{Journal of Modern Optics}, vol.~47,
  no. 2-3, pp. 355--376, 2000.

\bibitem{wei2003geometric}
T.-C. Wei and P.~M. Goldbart, ``Geometric measure of entanglement and
  applications to bipartite and multipartite quantum states,'' \emph{Physical
  Review A}, vol.~68, no.~4, p. 042307, 2003.

\bibitem{markham2007entanglement}
D.~Markham, A.~Miyake, and S.~Virmani, ``Entanglement and local information
  access for graph states,'' \emph{New Journal of Physics}, vol.~9, no.~6, p.
  194, 2007.

\bibitem{horodecki2005local}
M.~Horodecki, P.~Horodecki, R.~Horodecki, J.~Oppenheim, A.~Sen, U.~Sen,
  B.~Synak-Radtke, \emph{et~al.}, ``Local versus nonlocal information in
  quantum-information theory: formalism and phenomena,'' \emph{Physical Review
  A}, vol.~71, no.~6, p. 062307, 2005.

\bibitem{christandl2004squashed}
M.~Christandl and A.~Winter, ``“squashed entanglement”: an additive
  entanglement measure,'' \emph{Journal of mathematical physics}, vol.~45,
  no.~3, pp. 829--840, 2004.

\bibitem{gour2020optimal}
G.~Gour and M.~Tomamichel, ``Optimal extensions of resource measures and their
  applications,'' \emph{arXiv preprint arXiv:2006.12408}, 2020.

\bibitem{terhal2000schmidt}
B.~M. Terhal and P.~Horodecki, ``Schmidt number for density matrices,''
  \emph{Physical Review A}, vol.~61, no.~4, p. 040301, 2000.

\bibitem{gour2018monogamy}
G.~Gour and Y.~Guo, ``Monogamy of entanglement without inequalities,''
  \emph{Quantum}, vol.~2, p.~81, 2018.

\bibitem{vidal1999entanglement}
G.~Vidal, ``Entanglement of pure states for a single copy,'' \emph{Physical
  Review Letters}, vol.~83, no.~5, p. 1046, 1999.

\bibitem{virmani2000ordering}
S.~Virmani and M.~Plenio, ``Ordering states with entanglement measures,''
  \emph{Physics Letters A}, vol. 268, no. 1-2, pp. 31--34, 2000.

\bibitem{rains1997entanglement}
E.~M. Rains, ``Entanglement purification via separable superoperators,''
  \emph{arXiv preprint quant-ph/9707002}, 1997.

\bibitem{rains1999bound}
E.~M. Rains, ``Bound on distillable entanglement,'' \emph{Physical Review A},
  vol.~60, no.~1, p. 179, 1999.

\bibitem{rains2001semidefinite}
E.~M. Rains, ``A semidefinite program for distillable entanglement,''
  \emph{IEEE Transactions on Information Theory}, vol.~47, no.~7, pp.
  2921--2933, 2001.

\bibitem{eggeling2001distillability}
T.~Eggeling, K.~G.~H. Vollbrecht, R.~F. Werner, and M.~M. Wolf,
  ``Distillability via protocols respecting the positivity of partial
  transpose,'' \emph{Physical review letters}, vol.~87, no.~25, p. 257902,
  2001.

\bibitem{brandao2010reversible}
F.~G. Brandao and M.~B. Plenio, ``A reversible theory of entanglement and its
  relation to the second law,'' \emph{Communications in Mathematical Physics},
  vol. 295, no.~3, pp. 829--851, 2010.

\bibitem{yu2020quantifying}
D.-h. Yu, L.-q. Zhang, and C.-s. Yu, ``Quantifying coherence in terms of the
  pure-state coherence,'' \emph{arXiv preprint arXiv:2002.04330}, 2020.

\bibitem{jonathan1999minimal}
D.~Jonathan and M.~B. Plenio, ``Minimal conditions for local pure-state
  entanglement manipulation,'' \emph{Physical review letters}, vol.~83, no.~7,
  p. 1455, 1999.

\bibitem{nielsen1999conditions}
M.~A. Nielsen, ``Conditions for a class of entanglement transformations,''
  \emph{Physical Review Letters}, vol.~83, no.~2, p. 436, 1999.

\bibitem{bhatia2013matrix}
R.~Bhatia, \emph{Matrix analysis}.\hskip 1em plus 0.5em minus 0.4em\relax
  Springer Science \& Business Media, 2013, vol. 169.

\bibitem{vidal1999robustness}
G.~Vidal and R.~Tarrach, ``Robustness of entanglement,'' \emph{Physical Review
  A}, vol.~59, no.~1, p. 141, 1999.

\bibitem{vedral1998entanglement}
V.~Vedral and M.~B. Plenio, ``Entanglement measures and purification
  procedures,'' \emph{Physical Review A}, vol.~57, no.~3, p. 1619, 1998.

\bibitem{vedral2002role}
V.~Vedral, ``The role of relative entropy in quantum information theory,''
  \emph{Reviews of Modern Physics}, vol.~74, no.~1, p. 197, 2002.

\bibitem{buscemi2011entanglement}
F.~Buscemi and N.~Datta, ``Entanglement cost in practical scenarios,''
  \emph{Physical review letters}, vol. 106, no.~13, p. 130503, 2011.

\bibitem{regula2019one}
B.~Regula, K.~Fang, X.~Wang, and M.~Gu, ``One-shot entanglement distillation
  beyond local operations and classical communication,'' \emph{New Journal of
  Physics}, vol.~21, no.~10, p. 103017, 2019.

\bibitem{winter2016tight}
A.~Winter, ``Tight uniform continuity bounds for quantum entropies: conditional
  entropy, relative entropy distance and energy constraints,''
  \emph{Communications in Mathematical Physics}, vol. 347, no.~1, pp. 291--313,
  2016.

\bibitem{hayden2001asymptotic}
P.~M. Hayden, M.~Horodecki, and B.~M. Terhal, ``The asymptotic entanglement
  cost of preparing a quantum state,'' \emph{Journal of Physics A: Mathematical
  and General}, vol.~34, no.~35, p. 6891, 2001.

\bibitem{chitambar2020entanglement}
E.~Chitambar, J.~I. de~Vicente, M.~W. Girard, and G.~Gour, ``Entanglement
  manipulation beyond local operations and classical communication,''
  \emph{Journal of Mathematical Physics}, vol.~61, no.~4, p. 042201, 2020.

\bibitem{chen2010computation}
L.~Chen, A.~Xu, and H.~Zhu, ``Computation of the geometric measure of
  entanglement for pure multiqubit states,'' \emph{Physical Review A}, vol.~82,
  no.~3, p. 032301, 2010.

\bibitem{chen2011Connections}
L.~Chen, H.~Zhu, and T.-C. Wei, ``Connections of geometric measure of
  entanglement of pure symmetric states to quantum state estimation,''
  \emph{Physical Review A}, vol.~83, no.~1, p. 012305, 2011.

\bibitem{zhu2010additivity}
H.~Zhu, L.~Chen, and M.~Hayashi, ``Additivity and non-additivity of
  multipartite entanglement measures,'' \emph{New Journal of Physics}, vol.~12,
  no.~8, p. 083002, 2010.

\bibitem{briegel2001persistent}
H.~J. Briegel and R.~Raussendorf, ``Persistent entanglement in arrays of
  interacting particles,'' \emph{Physical Review Letters}, vol.~86, no.~5, p.
  910, 2001.

\bibitem{lu2009demonstrating}
C.-Y. Lu, W.-B. Gao, O.~G{\"u}hne, X.-Q. Zhou, Z.-B. Chen, and J.-W. Pan,
  ``Demonstrating anyonic fractional statistics with a six-qubit quantum
  simulator,'' \emph{Physical review letters}, vol. 102, no.~3, p. 030502,
  2009.

\bibitem{wieczorek2009experimental}
W.~Wieczorek, R.~Krischek, N.~Kiesel, P.~Michelberger, G.~T{\'o}th, and
  H.~Weinfurter, ``Experimental entanglement of a six-photon symmetric dicke
  state,'' \emph{Physical review letters}, vol. 103, no.~2, p. 020504, 2009.

\bibitem{wootters1998entanglement}
W.~K. Wootters, ``Entanglement of formation of an arbitrary state of two
  qubits,'' \emph{Physical Review Letters}, vol.~80, no.~10, p. 2245, 1998.

\bibitem{rossignoli2002generalized}
R.~Rossignoli and N.~Canosa, ``Generalized entropic criterion for
  separability,'' \emph{Physical Review A}, vol.~66, no.~4, p. 042306, 2002.

\bibitem{horodecki2001separability}
M.~Horodecki, P.~Horodecki, and R.~Horodecki, ``Separability of n-particle
  mixed states: necessary and sufficient conditions in terms of linear maps,''
  \emph{Physics Letters A}, vol. 283, no. 1-2, pp. 1--7, 2001.

\bibitem{luo2016general}
Y.~Luo, T.~Tian, L.-H. Shao, and Y.~Li, ``General monogamy of tsallis q-entropy
  entanglement in multiqubit systems,'' \emph{Physical Review A}, vol.~93,
  no.~6, p. 062340, 2016.

\bibitem{song2016general}
W.~Song, Y.-K. Bai, M.~Yang, M.~Yang, and Z.-L. Cao, ``General monogamy
  relation of multiqubit systems in terms of squared r{\'e}nyi-$\alpha$
  entanglement,'' \emph{Physical Review A}, vol.~93, no.~2, p. 022306, 2016.

\bibitem{coffman2000distributed}
V.~Coffman, J.~Kundu, and W.~K. Wootters, ``Distributed entanglement,''
  \emph{Physical Review A}, vol.~61, no.~5, p. 052306, 2000.

\bibitem{san2010tsallis}
J.~San~Kim, ``Tsallis entropy and entanglement constraints in multiqubit
  systems,'' \emph{Physical Review A}, vol.~81, no.~6, p. 062328, 2010.

\bibitem{bai2014general}
Y.-K. Bai, Y.-F. Xu, and Z.~Wang, ``General monogamy relation for the
  entanglement of formation in multiqubit systems,'' \emph{Physical review
  letters}, vol. 113, no.~10, p. 100503, 2014.

\bibitem{ou2007violation}
Y.-C. Ou, ``Violation of monogamy inequality for higher-dimensional objects,''
  \emph{Physical Review A}, vol.~75, no.~3, p. 034305, 2007.

\bibitem{guo2019monogamy}
Y.~Guo and G.~Gour, ``Monogamy of the entanglement of formation,''
  \emph{Physical Review A}, vol.~99, no.~4, p. 042305, 2019.

\end{thebibliography}
\end{document}